\documentclass[11pt]{article}
\RequirePackage[T1]{fontenc}
\RequirePackage[utf8]{inputenc}
\RequirePackage{amsmath,amssymb,amsthm}
\RequirePackage{graphicx}
\RequirePackage{textcomp}
\RequirePackage{etoolbox}
\RequirePackage{fullpage}
\RequirePackage{xcolor}
\RequirePackage{listings}
\RequirePackage{newfloat}
\DeclareFloatingEnvironment[
  fileext=lol,
  listname={List of Listings},
  name=Listing,
  placement=htbp
]{codelisting}
\RequirePackage[toc,page]{appendix}
\RequirePackage{centernot}

\newtheorem{proposition}{Proposition}

\RequirePackage{subcaption}
\RequirePackage{wasysym}
\RequirePackage{authblk}
\RequirePackage{bm}
\RequirePackage{bbm}
\RequirePackage{hyperref}
\RequirePackage{cleveref}
\RequirePackage{tikz}
\usetikzlibrary{backgrounds,positioning,arrows.meta,calc,decorations.pathmorphing,shapes.geometric,fit}
\RequirePackage[round]{natbib}
\RequirePackage{nicefrac}
\RequirePackage{enumitem}
\RequirePackage{booktabs}
\RequirePackage{float}

\tikzset{
  var/.style={draw, rounded corners, inner sep=4pt, align=center},
  exo/.style={var, fill=gray!12},
  endo/.style={var, fill=blue!6},
  det/.style={var, fill=green!6},
  arr/.style={-{Stealth[length=2.2mm]}, line width=0.8pt},
  lab/.style={font=\small}
}

\newtheorem{corollary}{Corollary}[section]
\theoremstyle{definition}
\newtheorem{definition}{Definition}[section]

\DeclareMathOperator{\E}{\mathbb{E}}

\DeclareMathOperator{\Var}{Var}
\DeclareMathOperator{\Cov}{Cov}

\title{Realizing Common Random Numbers: Event-Keyed Hashing for Causally Valid Stochastic Models}

\author[1,$\ast$]{Vince Buffalo}
\author[2]{Carl A. B. Pearson}
\author[1]{Daniel Klein}
\affil[1]{Institute for Disease Modeling\\Gates Foundation\\Seattle, WA}
\affil[2]{Department of Epidemiology\\University of North Carolina\\Chapel Hill, NC}
\affil[$\ast$]{Corresponding author: \href{mailto:vince.buffalo@gatesfoundation.org}{vince.buffalo@gatesfoundation.org}}

\begin{document}
\maketitle

\begin{abstract}

Agent-based models (ABMs) are widely used to estimate causal treatment effects
via paired counterfactual simulation. A standard variance reduction technique
is \emph{common random numbers} (CRNs), which couples replicates across
intervention scenarios by sharing the same random inputs. In practice, CRNs are
implemented by reusing the same base seed, but this relies on a critical
assumption: that the same draw index corresponds to the same modeled event
across scenarios. Stateful pseudorandom number generators (PRNGs) violate this
assumption whenever interventions alter the simulation's execution path,
because any change in control flow shifts the draw index used for all
downstream events. We argue that this \emph{execution-path-dependent} draw
indexing is not only a variance-reduction nuisance, but represents a
fundamental mismatch between the scientific causal structure ABMs are intended
to encode and the program-level causal structure induced by stateful PRNG
implementations. Formalizing this through the lens of structural causal models
(SCMs), we show that standard PRNG practices yield causally incoherent paired
counterfactual comparisons even when the mechanistic specification is otherwise
sound. We show that a remedy is to combine counter-based random number
generators (e.g., Philox/Threefry) with event identifiers. This decouples
random number generation from simulation execution order by making random draws
explicit functions of the particular modeled event that called them, restoring
the stable event-indexed exogenous structure assumed by SCMs.

\end{abstract}

\section{Introduction}
\label{sec:introduction}

Researchers increasingly use stochastic models to address \emph{counterfactual
questions} in epidemiology \citep{halloran2010design}, economics
\citep{tesfatsion2006handbook}, and policy analysis
\citep{epstein2006generative}. Such counterfactual questions ask ``what is the
outcome if we were to do intervention X'' and ``what is the treatment effect
between baseline and intervention?''. However, the \emph{fundamental problem of
causal inference} \citep{holland1986statistics} complicates our ability to
answer those questions: for each individual, we can only observe the outcome
under the treatment they actually received, never both potential outcomes. In
observational data, estimating causal treatment effects requires assumptions
about the data-generating process that are often unverifiable, such as no
unmeasured confounding (conditional exchangeability) across treatment groups
\citep{Hernan2024-jb,Greenland1986-pi}. Randomized trials estimate average
treatment effects (ATEs) by breaking confounding through random assignment of
individuals to treatment groups \citep{rubin1974estimating}, but
individual-level treatment effects remain unidentifiable since only one
potential outcome occurs. Moreover, randomization is often limited by ethical
or practical constraints.

By contrast, high-detail stochastic models such as agent-based models (ABMs)
simulate a fully specified data-generating process \emph{in silico}, allowing
modelers to \emph{intervene} on the simulated system and compare outcomes under
different interventions. For such ABMs, different samples yield different ATE
estimates, and researchers may additionally incorporate parameter uncertainty
via random sampling. Thus, analyses entail generating multiple simulation
replicates, much like empirical studies require multiple independent samples.
As in empirical experiments, the number of samples, choice of target
observable, and instrument noise together determine our ability to clearly
estimate effects.

A standard variance reduction technique is to couple stochastic replicates
\emph{across} intervention scenarios using \emph{common random numbers} (CRNs)
\citep{Stout2008-vn,klein2024noise}. Formally, CRNs construct a
\emph{probabilistic coupling} between baseline and intervention outcomes by
sharing the same random inputs (i.e., exogenous noises) across scenarios. Let
$Y^{(0)}$ and $Y^{(1)}$ denote an outcome under baseline and intervention, and
let $\Delta := \mathbb{E}[Y^{(1)}-Y^{(0)}]$ be the (model-implied) average
treatment effect. Given $M$ paired simulation replicates
$\{(Y^{(0)}_m,Y^{(1)}_m)\}_{m=1}^M$, the paired estimator
$\widehat{\Delta}=\frac{1}{M}\sum_{m=1}^M\!\big(Y^{(1)}_m-Y^{(0)}_m\big)$ has

\begin{equation}
\Var(\widehat{\Delta}) =\frac{1}{M}\Big(\Var(Y^{(1)})+\Var(Y^{(0)})-2\Cov(Y^{(1)},Y^{(0)})\Big),
\label{eq:ate-variance}
\end{equation}
so positive covariance induced by shared random inputs reduces Monte Carlo
variance relative to independent simulation runs (for which the covariance is
zero).

The stochastic simulation routines that underpin ABMs often use a
\emph{pseudorandom number generator} (PRNG), which is a deterministic algorithm
that produces a stream of numbers that approximate the properties of true
random numbers. Deterministic generators are critical to CRN approaches, as
well as for practical research objectives, such as analysis reproducibility.
Because these methods are deterministic, the PRNG can be thought of as a
preset sequence of values. Each draw moves to the next portion of this
sequence by advancing the generator's internal state.

In practice, a single centralized random number generator is often used for all
stochastic draws within a simulation replicate. That practice, combined with
the nature of PRNGs, leads to \emph{execution-path dependency} within a
simulation. That dependency can thwart attempts to use CRNs for variance
reduction.

Typically, modelers implement CRNs by reusing the same base seed (which sets
the PRNG state and initial position within the sequence) across intervention
scenario simulations. This approach guarantees an \emph{identical sequence} of
random draws $u_1, u_2, \ldots$ across all simulations seeded with the fixed
base seed. However, this approach relies on a critical assumption: the same
position in the random PRNG draw sequence corresponds to the same \emph{modeled
event} (i.e. a particular random variable) across coupled scenarios. Because
stateful PRNGs advance their internal state each function call (i.e., each
random draw), any simulation code that alters control flow changes the number
of prior draws consumed, and thus the draw index used for downstream events.
This \emph{execution-path-dependent} draw indexing decouples the modeled events
from the random draws they receive, violating the coupling assumption required
for CRNs.

For example, suppose a researcher wants to evaluate a vaccine intervention by
comparison to a baseline without vaccination. In the intervention scenario,
vaccine recipients stochastically avoid infection by drawing against their
efficacy. If that efficacy were 0 (i.e.\ a placebo), individuals do not
actually avoid infection. Yet the intervention scenario now includes an
additional draw for each vaccine recipient's efficacy check, shifting all
subsequent draws from the central PRNG. Downstream events---the newly infected
individual's incubation period, recovery time, contact events while infectious,
potential secondary transmissions---now receive different random numbers than
in the baseline. The non-placebo intervention of course has even more radical
shifts in execution path, since some infections (and their associated outcome
draws) are prevented entirely. In short, the requirement for actual CRNs is no
longer satisfied: the same modeled event (e.g. ``agent $i$ contacts agent $j$
at time $t$'') receives different random inputs depending on execution history,
even irrelevant history, rather than event identity alone.

The importance of matching random draws to the same modeled stochastic events
across intervention scenarios has been previously noted in individual-level
disease microsimulation \citep{Stout2008-vn}, in the context of ``perfect
counterfactuals'' for epidemic simulations \citep{Kaminsky2019-pc}, and
later in epidemiological ABM simulations \citep{klein2024noise}. A classic,
practical mitigation approach has been to use separate random number streams
for different event classes (e.g., transmission, demography)
\citep{Stout2008-vn,lecuyer2002object}. By partitioning events into independent
streams, changes in the number or timing of events in one class (e.g.,
preventing transmissions) do not immediately perturb the random draws consumed
by events in other classes (e.g., demographic processes). However, these
stream-based approaches are coarse: within a class, draw-index dependence
persists. Determining the appropriate granularity of stream partitioning
requires anticipating all possible execution path changes, making this
approach error-prone and difficult to validate for complex ABMs with
interdependent event types.

In this paper, we argue that execution-path-dependent draw indexing (inherent
in stateful PRNG usage) is not only a variance-reduction nuisance, but
represents a fundamental mismatch between the \emph{scientific} causal
structure ABMs are intended to encode, and the \emph{program-level} causal
structure induced by stateful PRNG implementations. For an ABM to function as a
valid structural causal model \citep{pearl2009causality} under interventions,
the identity of exogenous noise terms must be stable; only the structural
equations (or their inputs) should change. We formalize this problem through
the lens of structural causal models and show that standard PRNG practices can
yield causally incoherent paired counterfactual comparisons even when the
mechanistic specification is otherwise sound.

It is important to note that we are not claiming ABMs using stateful PRNGs are
invalid as \emph{probabilistic models} within a fixed intervention scenario.
For any fixed scenario $a$, the simulator induces a well-defined distribution
$P_a$ over outcomes across random seeds. Moreover, Monte Carlo replication
across independent seeds consistently estimates expectations under $P_a$.
Rather, the failure is in the \emph{across-scenario coupling}: when
interventions alter execution paths, even the same seed no longer produces
aligned noise terms across scenarios (because events are assigned noise by draw
index rather than event identity), making individual counterfactuals
ill-defined.

We show that a remedy to these problems created by stateful PRNGs is to combine
counter-based random number generators (e.g. Philox/Threefry,
\citealt{Salmon2011-zz}) with event identifiers, an approach we refer to as
\emph{event-keyed random number generation}. This approach decouples random
number generation from the simulation execution order by making random draws
explicit functions of the particular modeled event that called them. This
restores stable event-indexed exogenous structure assumed by structural causal
models, allowing ABMs to accurately represent the causal fidelity of the
encoded scientific model, maximize the variance-reduction benefits of CRN
approaches, and yield coherent counterfactual comparisons aligned with
SCM-style interventions.

\section{ABMs as Structural Causal Models}
\label{sec:abm-scm}

\subsection{Structural Causal Models Primer}

Structural Causal Models were introduced by
\citep{Pearl1995-ap,pearl2009causality,Greenland1999-yg} as a framework for
extending structural equation models (SEMs, \citealt{Wright1921-fp}), which
assumed linear relationships between variables, to arbitrary functional forms
with exogenous noise terms. An SCM is defined as a tuple $\mathcal{M} =
(\mathbf{U}, \mathbf{V}, \mathbf{F})$, where $\mathbf{U}$ is a set of exogenous
(i.e., external to the model) variables, $\mathbf{V}$ is a set of endogenous
(i.e., internal to the model) variables, and $\mathbf{F}$ is a set of functions
that determine the value of each endogenous variable as a function of its
parents in $\mathbf{V}$ and exogenous variables in $\mathbf{U}$. Often the
exogenous variables in $\mathbf{U}$ are random variables with specified
distributions $P(\mathbf{U})$, making the SCM probabilistic. When $\mathbf{U}$
contains specified values rather than random variables, these are typically treated
as model parameters. In Pearl's framework, the structural equations are
deterministic functions, so all randomness in the model is induced by the input
of exogenous random noise through the random variables.

SCMs are typically constructed so that exogenous variables in $\mathbf{U}$ are
mutually independent; dependencies between endogenous variables are explicitly
encoded in $\mathbf{F}$ rather than in $P(\mathbf{U})$ \citep[p. 44,
61-63][these are sometimes called \emph{Markovian} SCMs]{pearl2009causality}.
This simplifies reasoning through the directed acyclic graph (DAG) structure
induced by $\mathbf{F}$. In theory, $\mathbf{U}$ can be thought of as
independent uniform random draws that are then transformed by deterministic
functions to produce event outcomes; this coincides with the way
random draws are often made from PRNGs in stochastic algorithms (we note that
some distribution sampling algorithms, e.g.\ rejection sampling, consume a
variable number of underlying uniforms depending on distribution parameters,
which can create additional noise misalignment across scenarios unless the
underlying uniforms are themselves event-keyed).

Another key feature of SCMs is that they support \emph{interventions} through
the $do$-operator \citep{pearl2009causality}. An intervention $do(X=x)$ on
variable $X \in \mathbf{V}$ is modeled by replacing the structural equation for
$X$ in $\mathbf{F}$ with $X = x$, while holding all other exogenous variables
$\mathbf{U}$ and all other structural equations fixed. This operation creates a
new model variant $\mathcal{M}_{do(X=x)}$ for this intervention, which
represents the counterfactual ``what if we were to set $X$ to $x$?'' holding
all else equal. A key benefit of Pearl's SCM framework is that it provides a
formal language for reasoning about the effects of interventions through the
causal structure encoded in the DAG induced by $\mathbf{F}$.

Critically, because the exogenous variables $\mathbf{U}$ remain unchanged for
$do(X=x)$, interventions enable meaningful counterfactual comparisons: the same
``underlying randomness'' is used across different intervention scenarios,
allowing us to observe both potential outcomes for the same random realization
(i.e. same realizations of the random exogenous variables). This is precisely
what common random numbers in ABM simulations aim to achieve, but in practice
often fail to do so due to execution-path-dependent draw indexing when using
stateful PRNGs.

\subsection{ABMs as Scientific SCMs}
\label{sec:abm-as-scm}

ABMs can be interpreted as computational, dynamic implementations of
\emph{scientific} structural causal models. The mechanistic rules encoded in an
ABM correspond to the structural equations $\mathbf{F}$, the
internal state variables of agents and the environment correspond to endogenous
variables $\mathbf{V}$, and stochasticity enters through exogenous noise
variables $\mathbf{U}$. Outcomes of interest (e.g., incidence or
hospitalization time series) are components of $\mathbf{V}$ (or statistics
derived from them) that correspond to observables in the modeled system. When
we talk about the DAG implied by a particular ABM, we usually mean the
unrolled, time-expanded DAG that includes all time-indexed variables and
events.

While ABM algorithms encode a data-generating process, using them as causal
models that are capable of answering counterfactual questions requires
explicitly defining how exogenous noise terms $\mathbf{U}$ enter the model and
correspond to \emph{specific, unique} modeled stochastic events. This matters
because SCM-style interventions that compare baseline and intervention outcomes
must follow the semantics of the $do$-operator: when intervening on a variable,
we hold all exogenous noise terms fixed and only modify the structural
equations (or their inputs). Thus, for an ABM to function as a valid SCM under
interventions, the identity of the exogenous noise terms must remain stable
across scenarios.

For example, in a transmission model where infection status $I_i$ depends on
noise $U_i$ (e.g. a random draw from a Bernoulli distribution with some
infection probability $p_I$), a vaccine intervention should modify the
structural equation (reducing susceptibility via $p_I(A) = p_I (1-A\cdot VE)$
where $VE$ is the vaccine efficacy and $A \in \{0, 1\}$ is the vaccine
intervention indicator) but not change which noise term is used for agent $i$'s
infection. In \Cref{sec:stateful-prngs}, we show that standard, stateful PRNG implementations
violate this critical requirement. First, we define the conditions under which
an ABM algorithm is consistent with its intended scientific SCM.

In Pearl's framework, all counterfactuals live on a shared probability space
(defined by the probability measure $P(\mathbf{U})$ component of the SCM).
Thus, a \emph{world} is a realization of the entire exogenous collection
$\mathbf{U}$, while a \emph{scenario} (which we also refer to as a
\emph{counterfactual world}, e.g. baseline vs intervention) changes the
deterministic components of the SCM: the structural equations $\mathbf{F}$ or
their inputs, not $\mathbf{U}$ itself. For an ABM to respect this shared-world
semantics, it must ensure that for each intervention scenario $a$, the outcomes
$Y^{(a)}$ that the simulator outputs can be written as
\begin{equation} Y^{(a)} = \Phi_a(\mathbf{U}) \label{eq:shared-world}
\end{equation}
for the stochastic simulation routine $\Phi_a$ (composed of entirely
deterministic functions) with the \emph{same} $\mathbf{U}$ across all scenarios
$a$. The base seed $s$ then serves as an index to a particular world
$\mathbf{U} = G(s)$ for some deterministic function $G$. 

To make this concrete, consider a stochastic model of a singular transmission
event $e$ (i.e. ``does agent $i$ infect agent $j$ at time $t$?''). This event
depends on endogenous modeled state and inputs $S_e$ used to determine the risk
of this particular transmission event (e.g., who is in contact with whom,
contact duration, susceptibility and infectiousness modifiers, and any
intervention variables like vaccines). An ABM might implement this transmission rule as
\begin{equation}
    Y_e \;=\; \mathbf{1}\!\left[U_e < p_e(S_e)\right], \qquad U_e \sim \mathrm{Unif}(0,1),
\label{eq:event-shock}
\end{equation}
where $p_e(S_e)$ is fully determined by modeled mechanisms and modeled state,
and $U_e$ is the exogenous component that captures whatever the model leaves
outside $S_e$. This is precisely the potential outcomes framework
\citep{rubin1974estimating,rubin2005causal}: under an intervention scenario $a$
we would write $Y^{(a)}_e = \mathbf{1}\!\left[U_e <
p_e\big(S^{(a)}_e\big)\right]$, reflecting that the intervention acts on the
endogenous state or structural equations $S_e$ only, while the exogenous noise
$U_e$ is \emph{held fixed}. Then, the individual treatment effect 

\begin{align}
    \mathrm{ITE}_e &= Y^{(1)}_e - Y^{(0)}_e  \\
                &= \mathbf{1}\!\left[U_e < p_e\big(S^{(1)}_e\big)\right] - \mathbf{1}\!\left[U_e < p_e\big(S^{(0)}_e\big)\right]
\end{align}
is only a coherent counterfactual quantity because \emph{the same} $U_e$ is
used for this event for both potential outcomes $Y^{(0)}_e$ and $Y^{(1)}_e$.

Consequently, realizations of the exogenous noise $U_e \sim \mathrm{Unif}(0,1)$
partition the event outcomes into \emph{principal strata}
\citep{Frangakis2002-ie} for this event (akin to the unit- or individual-level
principal strata more common in randomized trials contexts). We show the
principal strata for event $e$ as in Table~\ref{tab:principal-strata}, where
$p_e^{(0)}$ and $p_e^{(1)}$ are the infection probabilities under baseline and
intervention scenarios respectively, with $p_e^{(1)} < p_e^{(0)}$ reflecting a
protective treatment. Extending this framing to multiple intervention
scenarios, with varying vaccine efficacy, the principal strata are obviously
ordered consistently with the vaccine efficacy. Practically speaking for this
example event, the identical $U_e$ guarantees that increasing vaccine efficacy
can only leave an outcome the same or convert to protection.

\begin{table}[h]
\centering
\begin{tabular}{llccl}
\toprule
Stratum & Condition & $Y_e^{(0)}$ & $Y_e^{(1)}$ & Interpretation \\
\midrule
Always-infected & $U_e < p_e^{(1)}$ & 1 & 1 & Infected regardless of intervention \\
Preventable & $U_e \in [p_e^{(1)}, p_e^{(0)})$ & 1 & 0 & Intervention prevents infection \\
Never-infected & $U_e \geq p_e^{(0)}$ & 0 & 0 & Not infected regardless \\
\bottomrule
\end{tabular}

\caption{Principal strata for a transmission event $e$ with baseline risk
$p_e^{(0)}$ and intervention risk $p_e^{(1)} < p_e^{(0)}$. The latent $U_e$
determines which stratum the event belongs to. With execution invariance, the
same $U_e$ is realized under both scenarios, making stratum membership
identifiable in simulation.}

\label{tab:principal-strata}
\end{table}

In observational studies, principal stratum membership is fundamentally
unidentifiable; only one potential outcome is observable per unit, so one
cannot determine whether a particular infection was ``preventable'' or
``always-infected.'' In theory, ABMs (and stochastic computational models more
generally) allow researchers to circumvent this limitation to enable event- or
unit-level causal reasoning; we can in principle use the same $U_e$ across
intervention scenarios to identify stratum membership for each event/unit.
Practically, this capability depends entirely on whether an event/unit receives
the same exogenous noise across scenarios. As we have illustrated, naive use of
stateful PRNG approaches easily leads to $U_e^{(0)} \neq U_e^{(1)}$ even when
using the same base seed, invalidating proper counterfactual comparisons.

\subsection{Execution Invariance for SCM-Consistent ABM Counterfactuals}
\label{sec:execution-invariance}

If researchers using ABMs wish to obtain causally valid answers to
counterfactual questions through simulation interventions, those ABMs must satisfy
structural requirements beyond simply encoding the mechanistic data-generating
process. In particular, the ``all-else equal'', shared-world semantics
described above require that the mapping between modeled event identity and
exogenous noise be \emph{invariant} to changes in execution history induced by
interventions. We call this property \emph{execution invariance}.

\begin{definition}[Execution Invariance] Given a base seed $s$ and considering
    an intervention variable $A \in \{0, 1\}$ (e.g., vaccination) that modifies
    the simulator's update rules (equivalently, structural equations
    $\mathbf{F}$), an ABM implementation is \emph{execution-invariant} with
    respect to a set of interventions if: for any two scenarios $a, a'$ and any
    stochastic event $e$ that occurs in both scenarios, the simulator uses the
    same exogenous noise value for $e$ in both runs when initialized with seed
    $s$. Formally, there exists a deterministic function $g$ such that:

$$U_e = g(s, \text{event\_id}_e)$$
where $\text{event\_id}_e$ is a canonical identifier for event $e$ that is 
stable across scenarios, not dependent on execution history.

\end{definition}

Circling back to the example of introducing a placebo vaccine: if a
``placebo'' scenario is constructed to be mechanistically identical to baseline
(no changes to state update rules or hazards, only additional draws), then any
divergence in outcomes indicates an execution-invariance violation, though
identical outcomes would not necessarily confirm satisfaction.

As noted for the non-placebo extension to that example, random events
may be \emph{conditionally occurring}; this does not violate the semantics of
the SCM framework, since $\mathbf{U}$ is the collection of noise terms for
\emph{all possible} events across all scenarios, and scenario-specific events
simply have no matched counterpart to compare against. For example, if a
vaccine saves an individual's life, that individual may go on to have offspring
whose lifetimes then require exogenous noise draws conditional on the event
that their parent survived; these events have no counterpart in the scenario
where the parent died. We thus distinguish \emph{conditionally occurring}
events from \emph{noise contamination} for an event. The problem arises when
events that \emph{do} occur in both scenarios receive different noise due to
draw-index coupling (i.e. random draw consumption differences): here the
counterfactual comparison becomes incoherent because we are comparing the same
modeled event under different exogenous realizations.

Execution invariance embodies the \emph{modularity} principle in causal
modeling \citep{pearl2009causality,Hausman1999-xo}: structural mechanisms
should be independent components that can be modified under interventions
without affecting others. We apply this principle specifically to exogenous
noise terms: each event's noise $U_e$ should be an autonomous component whose
identity does not depend on whether other events occur. In standard (Markovian)
SCMs, exogenous variables are often independent by construction.

% This shared-world property implies what we might call a \emph{commutation} 
% requirement: sampling a world (via seed $s$) and then intervening should be 
% equivalent to intervening and then sampling, because interventions change only 
% mechanisms, not which exogenous variables exist or which events they correspond 
% to. For seed-matched paired runs, this means
% %
% \begin{equation}
% \big(Y^{(0)}(s), Y^{(1)}(s)\big) = \big(\Phi_0(G(s)), \Phi_1(G(s))\big)
% \label{eq:paired-runs}
% \end{equation}
% %
% rather than having separate "exogenous tapes" $G_0(s)$ and $G_1(s)$ for each 
% scenario. Without execution invariance, interventions effectively re-index 
% which exogenous variables feed which events, breaking this shared-world 
% structure and rendering individual counterfactuals ill-defined—even when 
% population-level estimates remain unbiased.

Importantly, we note that execution invariance is stronger than within-scenario
distributional correctness. Most forward simulations built with stateful PRNG
produce correct marginal and joint distributions over outcomes within a fixed
scenario. Execution invariance additionally requires that seed-matched paired
runs correspond to an SCM-style counterfactual coupling in which interventions
change only structural equations (or their inputs), not the identity of
event-indexed noise terms. As we show in \Cref{sec:stateful-prngs}, standard
stateful PRNG implementations violate this requirement by making noise identity
depend on endogenous execution history via the draw index.

\section{Stateful PRNGs and Execution-Path-Dependent Draw Indexing}
\label{sec:stateful-prngs}

A \emph{pseudorandom number generator} (PRNG) is a deterministic algorithm that
produces a sequence of reproducible, approximately random numbers from an
initial \emph{seed}. \emph{Stateful} PRNG algorithms (e.g. Mersenne Twister,
PCG, xoshiro) can be thought of as a function $(u_k, s_{k+1}) = F(s_k)$ that
are seeded with some initial state $s_0$, and for draw $k$, return a random
uniform $u_k$ and new state $s_{k+1}$ that is then set as the internal state
for the next draw.

\subsection{Stateful PRNGs as Execution-Path-Indexed Noise}

Because stateful PRNGs advance mutable internal state on every call, they
exhibit execution-path-dependent draw indexing: the particular
random draw (i.e. outcome of a random event) is not determined solely by that
event's \emph{model inputs} (i.e. parameters or other random variables), but
also by how many PRNG draws have been consumed earlier in the execution (i.e.,
by the PRNG's current internal state). As a result, changes in control flow
(e.g.\ branching or early exits) that consume a different number of draws can
change downstream outcomes even when the downstream event's inputs are
unchanged. Overall this means that the identity of the exogenous noise term
used for a particular modeled event is not indexed by the event itself, but
rather by the number of prior draws consumed (the \emph{draw index}).
Critically, any time the program control flow (i.e. execution path) changes the
number of draws consumed before an event across intervention scenarios, the
identity of the exogenous noise terms used for modeled events also changes,
violating execution invariance. We illustrate this violation with a concrete
example.

\subsection{Example: Infection Model with Incubation Time Draw}

To illustrate this problem concretely, consider a toy infection model where
each of $N$ individuals has some base probability of infection
\texttt{p\_infect[i]} on day 0, and if infected, an incubation time is drawn
from some distribution via a function \texttt{draw\_incubation(rng)}. We
consider two scenarios: a baseline without vaccination (\texttt{vaccinated =
False}), and an intervention scenario where only person 1 is vaccinated
(\texttt{vaccinated = True}), reducing their probability of infection by some
vaccine efficacy factor $VE$. 

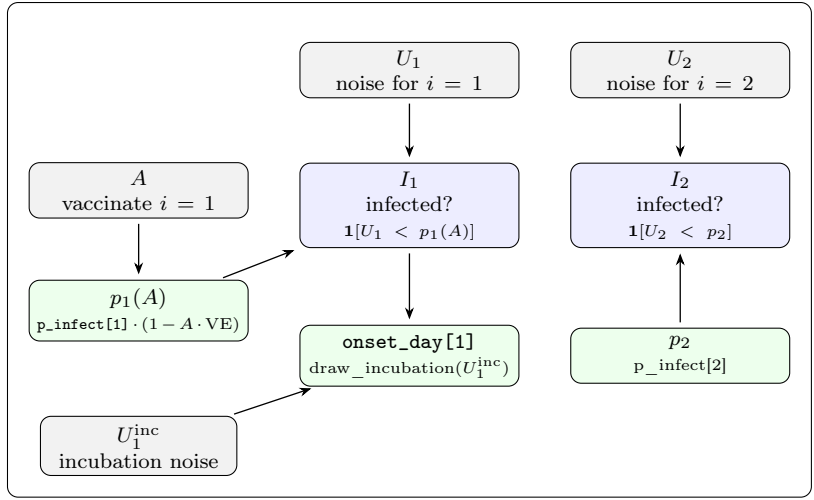
\begin{figure}[!htbp]
\centering
\begin{tikzpicture}[
  font=\scriptsize,
  var/.style={draw, rounded corners, inner sep=3pt, align=center, text width=2.7cm},
  exo/.style={var, fill=gray!10},
  endo/.style={var, fill=blue!7},
  det/.style={var, fill=green!7},
  arr/.style={-{Stealth[length=1.6mm]}, line width=0.5pt, shorten >=2pt, shorten <=2pt},
]

% Title anchor so fit includes title space
\node[draw=none, inner sep=0pt, outer sep=0pt] (titleanchor) at (0,3.2) {};

% Exogenous noises (event-indexed; fixed identity across counterfactuals)
\node[exo] (U1) at (3.6,2.6) {$U_1$\\noise for $i=1$};
\node[exo] (U2) at (7.2,2.6) {$U_2$\\noise for $i=2$};

% Intervention and risk for person 1
\node[exo] (A)  at (0,1.0) {$A$\\vaccinate $i=1$};
\node[det] (p1) at (0,-0.6) {$p_1(A)$\\{\tiny $\texttt{p\_infect[1]}\cdot(1-A\cdot\mathrm{VE})$}};
\draw[arr] (A) -- (p1);

% Person 1 infection uses U1 (not a shared stream)
\node[endo] (I1) at (3.6,0.8)
  {$I_1$\\infected?\\{\tiny $\mathbf{1}[U_1<p_1(A)]$}};
\draw[arr] (p1) -- (I1);
\draw[arr] (U1) -- (I1);

% Scheduling variable (onset day) drawn only if infected; does not affect other hazards here
\node[det] (onset) at (3.6,-1.2)
  {\texttt{onset\_day[1]}\\{\tiny draw\_incubation$(U^{\text{inc}}_1)$}};

% Optional separate noise for incubation (keeps semantics clean)
\node[exo, text width=2.4cm] (Uinc1) at (0.0,-2.4)
  {$U^{\text{inc}}_1$\\incubation noise};

\draw[arr] (I1) -- (onset);
\draw[arr] (Uinc1) -- (onset);

% Person 2 risk and infection uses its own noise U2
\node[det] (p2) at (7.2,-1.2) {$p_2$\\{\tiny p\_infect[2]}};
\node[endo] (I2) at (7.2,0.8)
  {$I_2$\\infected?\\{\tiny $\mathbf{1}[U_2<p_2]$}};

\draw[arr] (p2) -- (I2);
\draw[arr] (U2) -- (I2);

% (No arrows from onset_day[1] to I2 in this simplified scientific SCM.)

% Panel box + title
\node[
  draw, rounded corners, inner sep=8pt,
  fit=(titleanchor)(A)(p1)(U1)(I1)(Uinc1)(onset)(U2)(I2)(p2),
  label={[anchor=south west, font=\small]north west:{\textbf{Scientific SCM:} event-indexed noise (stable across counterfactuals)}}
] (box) {};

\end{tikzpicture}
\caption{Intended scientific causal structure. Each infection event has
event-indexed exogenous noise ($U_1$, $U_2$) with stable identity across
counterfactuals. The vaccination intervention $A$ affects only person 1's risk
$p_1$. Downstream scheduling (\texttt{onset\_day}) does not create causal paths
to unrelated infections. \textbf{Colors:} \colorbox{gray!10}{gray} = exogenous
inputs/interventions; \colorbox{green!7}{green} =
deterministic/bookkeeping/scheduling; \colorbox{blue!7}{blue} = model outcomes
(infections).}
\label{fig:scientific-scm}
\end{figure}

Since this models only events within a single timestep, infections are causally
independent events across individuals. The intended scientific SCM in this case
has causally independent infection events, as shown in
Figure~\ref{fig:scientific-scm} by the lack of a causal path between $I_1$ and
$I_2$. Consequently, the exogenous noise terms $U_1$ and $U_2$ represent
autonomous random draws for individuals 1 and 2, and are thus invariant across
scenarios in which the base seed fixes the random stream. Under the intended
SCM, person 2's infection depends only on their own risk $p_2$ and noise $U_2$
--- there is no causal pathway from person 1's infection status.

A typical implementation of this model using stateful PRNGs appears in
Listing~\ref{lst:stateful-prng}. The critical problem emerges in the
conditional \texttt{if} block that draws incubation time \emph{only} when
infection occurs: this draw advances the PRNG state even though it serves only
to schedule a future event and has no mechanistic effect on other individuals'
infection risks.

Consider execution under both scenarios with a fixed seed $s$ that generates
the random draw sequence $R_1, R_2, R_3, \ldots$. Suppose that under the
baseline scenario with seed $s$, person 1 becomes infected. In that case, the
PRNG calls proceed as: $R_1 \to$ person 1's infection draw, $R_2 \to$ person
1's incubation time, and $R_3 \to$ person 2's infection test. But now suppose
that under the vaccination intervention, the same initial draw $R_1$ that
infected person 1 in the baseline scenario \emph{prevents} person 1's infection
under the intervention. Then, the \texttt{draw\_incubation()} call never
executes, so the PRNG call sequence becomes: $R_1 \to$ person 1's infection
draw, and $R_2 \to$ person 2's infection draw. Person 2's infection now uses
\emph{different} exogenous noise ($R_2$ versus $R_3$) across scenarios, purely
due to draw-index shifts rather than any change in their infection mechanism or
risk.

\begin{codelisting}[!htbp]
\centering
\begin{lstlisting}
# Input: p_infect[1..N]
# Parameters: VE

# Same seed s for baseline and vaccine run
function simulate_infections(seed s, vaccinated):
  # Initialize a stateful PRNG with seed s
  rng = PRNG(seed = s)

  # Track cumulative cases in this run
  cases = 0

  for i in 1..N:
      # By default, use person i's base probability of infection
      p = p_infect[i]

      if (i == 1) and vaccinated:
          # Intervention: vaccinate one person, reducing their
          # susceptibility (scales infection risk by 1 - VE)
          p = p * (1 - VE)

      # Bernoulli draw for whether person i becomes infected
      infected = (rng.rand() < p)

      if infected:
          # If infected, increment case count
          cases += 1

          # This draw ONLY happens if infected, consuming one PRNG call
          # In vaccine scenario, this may not execute, shifting all
          # subsequent draws by one position
          onset_day[i] = draw_incubation(rng)

  return cases
\end{lstlisting}
\caption{Stateful PRNG implementation. The conditional draw for incubation time
advances the PRNG state, shifting the draw index used for downstream events.}
\label{lst:stateful-prng}
\end{codelisting}

Perhaps surprisingly, this \emph{execution-path-dependent draw indexing}
implies a fundamentally different causal structure than the scientific model
intended. To represent what the code actually computes as a SCM, we must
introduce a new endogenous variable $K_2$ that tracks which random draw index
is used for person 2's infection event. This variable has no scientific meaning
--- it is purely an implementation artifact. Moreover, it creates a spurious
causal pathway from person 1's infection to person 2's infection that violates
the intended model structure, as shown in Figure~\ref{fig:option2-scm}.

\begin{figure}[H]
\centering
\begin{tikzpicture}[
  font=\scriptsize,
  var/.style={draw, rounded corners, inner sep=3pt, align=center, text width=2.6cm},
  exo/.style={var, fill=gray!10},
  endo/.style={var, fill=blue!7},
  det/.style={var, fill=green!7},
  arr/.style={-{Stealth[length=1.6mm]}, line width=0.5pt, shorten >=2pt, shorten <=2pt},
]

% Title anchor so fit includes title space
\node[draw=none, inner sep=0pt, outer sep=0pt] (titleanchor) at (0,3.4) {};

% Top row: seed -> draw stream
\node[exo] (S0) at (0,2.8) {$S_0$\\seed};
\node[det] (R1) at (3.6,2.8) {$R_1$\\draw \#1};
\node[det] (R2) at (7.2,2.8) {$R_2$\\draw \#2};
\node[det] (R3) at (10.8,2.8) {$R_3$\\draw \#3};

\draw[arr] (S0) -- (R1);
\draw[arr] (R1) -- (R2);
\draw[arr] (R2) -- (R3);

% Intervention and risk for person 1
\node[exo] (A)  at (0,1.0) {$A$\\vaccinate $i=1$};
\node[det] (p1) at (0,-0.6) {$p_1(A)$\\{\tiny $\texttt{p\_infect[1]}\cdot(1-A\cdot\mathrm{VE})$}};
\draw[arr] (A) -- (p1);

% Person 1 infection uses draw #1
\node[endo] (I1) at (3.6,0.8)
  {$I_1$\\infected?\\{\tiny $\mathbf{1}[R_1<p_1(A)]$}};
\draw[arr] (p1) -- (I1);
\draw[arr] (R1) -- (I1);

% Onset day draw - only if I1=1, consumes R2 (green = bookkeeping/scheduling)
\node[det, text width=6cm] (onset) at (6,-1.4)
{\texttt{onset\_day[1]}{\tiny $:= \begin{cases} \text{draw\_incubation}(R_2) & \mathrm{if} \; I_1{=}1 \\ \bot\text{ (undefined)} & \mathrm{if} \; I_1{=}0 \end{cases}$}};
\draw[arr] (I1) -- (onset.100);
\draw[arr] (R2) -- ++(0,-0.8) -- ++(-1.9,0) -- ++(0,-1.8) -| (onset.north);

% K2 = draw index for person 2
\node[det, text width=2.8cm] (K2) at (7.2,1.1)
  {$K_2$\\draw index for $i=2$\\{\tiny $K_2 = 2 + I_1$}};
\draw[arr] (I1) -- (K2);

% Person 2 infection uses R_{K2}
\node[endo, text width=2.8cm] (I2) at (10.8,0.6)
  {$I_2$\\infected?\\{\tiny $\mathbf{1}[R_{K_2}<p_2]$}};
\node[det] (p2) at (10.8,-1.2) {$p_2$\\{\tiny p\_infect[2]}};

\draw[arr] (p2) -- (I2);
\draw[arr] (K2) -- (I2);
\draw[arr] (R2) to[out=-10,in=150] (I2);
\draw[arr] (R3) -- (I2);

% Panel box + title
\node[
  draw, rounded corners, inner sep=8pt,
  fit=(titleanchor)(S0)(R3)(A)(p1)(I2)(onset),
  label={[anchor=south west, font=\small]north west:{\textbf{Program-level SCM:} draw-index $K_2$ made explicit}}
] (box) {};

\end{tikzpicture}
\caption{Actual program-level causal structure induced by stateful PRNG. The
draw index $K_2$ for person 2 becomes endogenous, depending on whether person 1
was infected ($K_2 = 2 + I_1$). If person 1 gets infected, the
\texttt{draw\_incubation} call consumes $R_2$, forcing person 2 to use $R_3$;
otherwise person 2 uses $R_2$. This creates a spurious causal pathway ($I_1 \to
K_2 \to I_2$) absent from the scientific model, violating execution invariance.
\textbf{Colors:} \colorbox{gray!10}{gray} = exogenous inputs/interventions;
\colorbox{green!7}{green} = deterministic/bookkeeping/scheduling;
\colorbox{blue!7}{blue} = model outcomes (infections).}
\label{fig:option2-scm}
\end{figure}
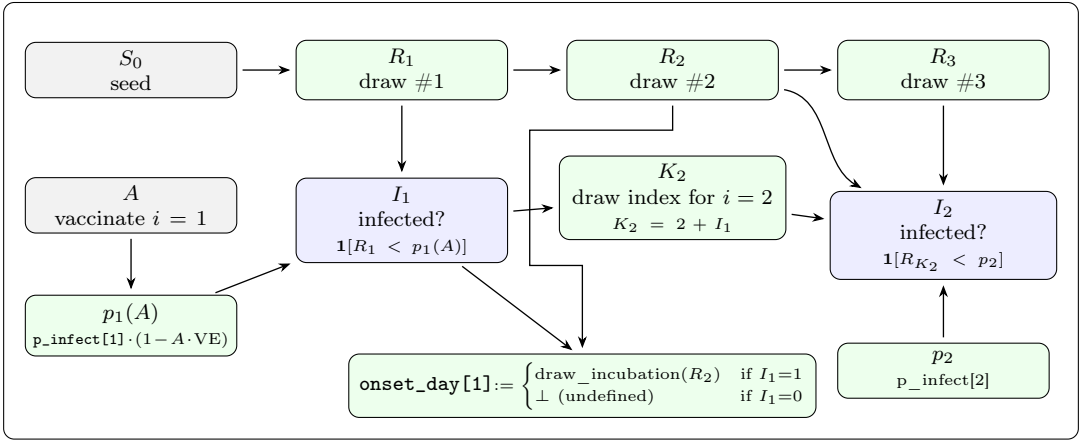

Formally, the intended scientific SCM specifies infection events as:

\begin{equation}
I_1 = \mathbf{1}[U_1 < p_1(A)], \qquad I_2 = \mathbf{1}[U_2 < p_2],
\label{eq:scientific-scm}
\end{equation}
where $U_1, U_2$ remain fixed across scenarios; only the structural equation for
$p_1(A)$ changes under intervention. By contrast, the program-level SCM induced by stateful
PRNG instead implements:

\begin{equation}
I_1 = \mathbf{1}[R_1 < p_1(A)], \qquad K_2 = 2 + I_1, \qquad 
I_2 = \mathbf{1}[R_{K_2} < p_2].
\label{eq:program-scm}
\end{equation}
The draw index $K_2$ is now endogenous, making the noise identity for person
2's infection depend on person 1's outcome. This violates execution invariance:
the same modeled event (person 2's infection) receives different exogenous
noise depending on draw index (e.g. PRNG consumption history) rather than event
identity alone.

For this toy example, one might readily propose draw schemes that avoid the
problem. However, those schemes rapidly increase in complexity with the model,
and execution invariance becomes a matter requiring both non-trivial proof that
an approach provides that guarantee and careful implementation.

\subsection{Stateful PRNGs Violate Execution Invariance}

The toy infection model above illustrates how stateful PRNGs create
execution-dependent noise assignment that implies an entirely different causal
model than the scientific one intended. We now formalize this violation to
characterize when and why it occurs.

\begin{proposition}[Violation of Execution Invariance]
\label{prop:stateful-violation}

Consider an ABM implementation using stateful PRNG with seed-matched
initialization across intervention scenarios $a$ and $a'$. If an event $e$
occurs in both scenarios and the number of PRNG calls (within the relevant
random stream, if multiple streams are used) before $e$ differs between
scenarios, then:

\begin{enumerate}[nosep]
\item The draw index $K_e$ becomes endogenous (depends on prior outcomes)
\item Event $e$ receives different noise: $U_e^{(a)} \neq U_e^{(a')}$
\item The implementation violates execution invariance
\end{enumerate}
\end{proposition}

\begin{proof}[Proof sketch]
Let $n^{(a)}$ and $n^{(a')}$ denote the number of PRNG calls before 
event $e$ in scenarios $a$ and $a'$ respectively. With stateful PRNG, 
event $e$ receives draw $R_{n^{(a)}+1}$ in scenario $a$ and draw 
$R_{n^{(a')}+1}$ in scenario $a'$. If $n^{(a)} \neq n^{(a')}$ (which 
occurs when interventions change conditional execution), then event $e$ 
receives different stream positions, violating $U_e = g(s, \text{event\_id}_e)$. 
\end{proof}

\begin{corollary}
Seed-matched CRNs using stateful PRNGs fail to produce valid
counterfactual couplings when interventions alter execution paths such that
draw index decouples from event identity.
\end{corollary}

These violations have consequences for statistical efficiency (variance
reduction becomes unpredictable or even negative), counterfactual coherence
(individual-level treatment effects become ill-defined), and auxiliary
analyses such as sensitivity analysis and mediation analysis. We detail these
consequences in \Cref{app:consequences}.

\section{Event-Keyed Random Number Generators}
\label{sec:event-keyed}

In Section~\ref{sec:execution-invariance} we defined execution invariance as a
stable mapping from event identity to exogenous noise across intervention
scenarios and random seeds: $U_e = g(s, \text{event\_id}_e)$. Note that this
requires the modeler to define what are \emph{matched events} across worlds,
which is a critical part of defining the scientific SCM; event matching, while
conceptually simple, can be challenging in practice as we illustrate in
Section~\ref{ssec:event-key-identifiers}.

Stateful PRNGs typically violate the execution invariance property by making
noise assignment depend on draw index rather than event identity (i.e. events
are not matched). We now present a conceptual approach that restores execution
invariance by making random draws explicit functions of stable event
identifiers.

\subsection{Counter-Based Random Number Generators}

Unlike stateful PRNGs that maintain a mutable internal state,
\emph{counter-based PRNGs} (or sometimes known as \emph{hash-based PRNGs}) are
purely functional algorithms \citep{Hughes1989-ps} whose output only depends on
the input arguments; these functions have no mutable internal state or side
effects. Counter-based PRNGs can be thought of as functions of the form

$$ R = g(\text{key}, \text{counter}) $$
where $g$ is a block-cipher-like mixing function, \texttt{key} is a fixed seed
(or stream identifier), and \texttt{counter} indexes which variate is returned.
For event-keyed usage, the seed serves as the key (selecting a ``world''), and
the counter is replaced by a composite event identifier that encodes both
the event identity and, when an event requires multiple draws, a within-event
draw index. That is, we call $g(\text{seed}, \text{event\_id}_e)$ where
$\text{event\_id}_e$ uniquely identifies each draw. When an event requires
multiple independent variates (e.g.\ both a Bernoulli trial and a duration),
each draw uses a distinct event identifier (e.g.\ by appending an index or
using a separate event type label).

Because counter-based PRNGs are purely functional, different inputs produce
values that have the properties of independent random draws, and
repeated calls with the same inputs always return the same output
(i.e. enabling random access); naturally, this enables execution invariance by
making random draws explicit functions of event identity. Historically,
counter-based PRNGs were designed to enable random number generation in
parallel computing contexts \citep{Salmon2011-zz}. For our purposes, they
decouple random draws from execution order by making each draw an explicit
function of a stable identifier rather than an implicit function of execution
history.

Listing~\ref{lst:event-keyed} shows a concrete implementation of event-keyed
randomness applied to toy infection example.

\begin{codelisting}[H]
\centering
\begin{lstlisting}
# Input: p_infect[1..N]
# Parameters: VE

# Same seed s for baseline and vaccine run
function simulate_infections_event_keyed(seed s, vaccinated):
    # No mutable PRNG state
    cases = 0
    for i in 1..N:
        p = p_infect[i]

        if (i == 1) and vaccinated:
            p = p * (1 - VE)

        # Event-specific key ensures same draw for same event
        event_key = hash("infection", i)
        infected = (CBRNG(seed, event_key) < p)

        if infected:
            cases += 1
            # Uses different key - doesn't shift other events
            inc_key = hash("incubation", i)
            onset_day[i] = draw_incubation(CBRNG(seed, inc_key))

    return cases
\end{lstlisting}
\caption{Event-keyed implementation. Each event uses a unique identifier combining
event type and agent ID, passed alongside the seed to a counter-based PRNG. Conditional execution (\texttt{onset\_day} draw)
no longer shifts downstream random numbers. Person 2's infection always uses
the same key regardless of person 1's outcome.}
\label{lst:event-keyed}
\end{codelisting}

\subsection{Model Building and Event-Key Identifiers}
\label{ssec:event-key-identifiers}
While counter-based PRNGs provide the building blocks for execution invariance,
they require the modeler to define \emph{event identifiers} that remain stable
across intervention scenarios. Defining a stable event-key identity across
counterfactual worlds is therefore equivalent to defining what constitutes
``the same event'' across counterfactual worlds. We emphasize that this is a
substantive \emph{modeling choice} and \emph{modeling assumption}, and not
automatable (cf.\ the debate over ``transworld identity'' and Lewis's
\citeyear{Lewis1986-ip} counterpart-theoretic alternative; see
\citealt{sep-identity-transworld} and references therein).

Scientific and statistical models routinely distinguish what is represented
explicitly in the model, from what is unexplained mechanism or process
stochasticity treated as uncertainty in the model (in a regression context one
might loosely call this ``signal'' versus ``noise''). Importantly, SCMs make a
related boundary explicit: what is included in the modeled state and structural
mechanisms versus what enters as exogenous ``chance''. Event-keyed randomness
adds a further modeling choice about the \emph{structure} of that exogenous
chance \emph{across counterfactual worlds}: it specifies which exogenous random
variation is held fixed across scenarios for events we regard as identical.

Recall from Section~\ref{sec:abm-as-scm} that transmission events take the form
$Y_e = \mathbf{1}[U_e < p_e(S_e)]$, where $S_e$ captures the full modeled state
and inputs used to determine the risk of transmission event $e$, and $U_e \sim
\mathrm{Unif}(0,1)$ is the exogenous noise. Under SCM semantics, interventions
change the modeled mechanisms that determine $p_e(S_e)$ (through changes in
$S_e$ or the functional form of $p_e$), while holding the exogenous noise
collection $\mathbf{U}$ fixed. An event-key scheme therefore determines which
particular exogenous variable $U_e$ is associated with a given modeled
transmission occurrence, and thus what it means to hold ``all else equal'' when
comparing across counterfactual worlds. Given that execution invariance
requires $U_e = g(s, \text{event\_id}_e)$, the event-key design is precisely
what defines which events are considered the same across counterfactual worlds.

We motivate how different event-key designs correspond to different modeling
choices about the structure of exogenous randomness across counterfactual
worlds by adding some realism and detail to our prior transmission model.
Consider modeling infection events between healthcare workers and patients in a
healthcare clinic. A susceptible patient $i$ has a scheduled close-contact
interaction at time $t$ (e.g. triage or examination). Suppose in the baseline
scenario, the healthcare worker present for the encounter is individual $j$,
who is infectious; under another scenario (e.g. because $j$ is absent due to
earlier infections or policy), the worker present at the same time $t$ is
instead a \emph{different} healthcare worker $k$.

In both scenarios, suppose these encounters are identical with respect to the
determinants of transmission risk that enter $p_e(\cdot)$ (e.g.
role, duration, PPE, infectiousness, etc.), so that the modeled risk is exactly
the same, even though the partner identity differs. Then, the event-key design
determines the counterfactual semantics of the infection event: should the
residual exogenous chance variable $U_e$ be attached to the \emph{contact
opportunity} at time $t$ (and thus shared across scenarios), or attached to the
\emph{dyad} (patient--worker pair) (and thus change when the worker changes)?

One coherent choice is to treat the event as the \emph{contact opportunity}
itself, with contact-partner identity (and any partner-specific covariates)
entering only through the modeled state $S_e$. We call this \emph{slot-keyed}
(we note that this is distinct from the ``slotting'' terminology used in
\citealt{klein2024noise}) event identity: the slot-based event-key encodes that
$U_e = U_{i,r}$ is the residual chance associated with patient $i$'s $r$th
scheduled encounter on day $t$,
\begin{equation}
    \text{event\_id}_e = (t,i,r), \qquad U_{i,r} = g\!\big(s,(t,i,r)\big).
\end{equation}
Under slot-keying, the transmission potential outcomes \emph{share the same
exogenous noise} as long as the same contact opportunity occurs at time $t$,
regardless of whether the contact partner is $j$ or $k$:
\begin{equation}
    Y_e^{(0)} = \mathbf{1}[U_{i,r} < p_e(S^{(0)}_e)], \qquad Y^{(1)}_e = \mathbf{1}[U_{i,r} < p_e(S^{(1)}_e)].
\end{equation}
The identity of the healthcare worker affects transmission only through
$p_e(S_e)$, not through the exogenous noise. The counterfactual question
answered is: ``does patient $i$ become infected at their $r$th encounter,
given the modeled determinants of that encounter?''

A different coherent choice is to treat the event as inherently \emph{dyadic},
so that the residual chance is attached to a specific patient--worker
interaction. Under \emph{dyad-keyed} event identity, the change in partner
implies different exogenous noise:
\begin{equation}
    \text{event\_id}_e = (t,i,j), \qquad U_{i,j} = g\!\big(s,(t,i,j)\big).
\end{equation}
Consequently, changing the worker from $j$ to $k$ means the potential
outcomes of the transmission event across intervention scenarios use 
\emph{different exogenous noise}
\begin{equation}
    Y^{(0)}_e = \mathbf{1}[U_{i,j} < p_e(S^{(0)}_e)], \qquad Y^{(1)}_e = \mathbf{1}[U_{i,k} < p_e(S^{(1)}_e)]
\end{equation}
where $U_{i,j} \neq U_{i,k}$ are distinct independent draws. This corresponds
to a different counterfactual semantics: ``holding fixed the latent chance
associated with exposure to a \emph{particular individual},'' rather than
holding fixed the latent chance associated with a \emph{particular contact
opportunity} slot. This counterfactual question at the event- or
individual-level is only well-posed for comparisons involving the same worker;
when the healthcare worker that patient $i$ interacts with changes from $j$ to $k$, no valid
counterfactual exists, and thus no individual- or event-level
treatment effect can be estimated. Only when we average over many such events
(e.g. to estimate ATEs) can we obtain well-defined
contrasts.

Fundamentally, the choice between slot-keying and dyad-keying reflects an
exchangeability assumption. Slot-keying assumes that potential outcomes under
different contact partners share the same exogenous noise: the residual
variation in transmission risk is a property of the contact \emph{opportunity},
not the specific partner. This is appropriate when partners are exchangeable
conditional on $S_e$; that is, when all relevant partner-specific heterogeneity
is captured in deterministically modeled state. Dyad-keying relaxes this
assumption, allowing for partner-specific variation not captured in $S_e$;
implying that valid individual-level counterfactuals require the same partner
in both scenarios. This parallels exchangeability assumptions in causal
inference \citep{Greenland1986-pi}: just as conditional exchangeability
requires that treatment groups be comparable after adjusting for covariates,
slot-keying assumes contact partners are interchangeable after conditioning on
modeled state $S_e$.

This illustrates why event identity (and thus event-key design) is not
automatable: it takes a modeling decision to explicitly specify what remains
exogenous across worlds, i.e., which aspects of a modeled stochastic occurrence
are treated as ``the same underlying randomness'' after accounting for modeled
mechanisms. Stateful PRNGs implicitly set event identity to draw index, thereby
implicitly (and often unintentionally) choosing a transworld identity map that
declares events identical only until a draw-index shift occurs. By contrast,
event-keyed randomness makes that choice explicit and therefore an inspectable
(and scrutable) model choice. 

\subsection{Practical Event-Key Design Guidelines}

Three practical principles guide event-key design. First, event identifiers
must be sufficiently \emph{granular} such that distinct modeled events receive
distinct exogenous draws; otherwise, the modeler introduces spurious
dependencies between events that should be conditionally independent given the
modeled state. For example, omitting the time index and using only agent ID for
daily infection draws actually means there is only one event, corresponding to
e.g.\ setting the agent's lifetime susceptibility. This choice introduces
dependence between what scientifically should perhaps be independent trials.

Two distinct concerns arise regarding key uniqueness. First, \emph{semantic
reuse}: querying the same event key multiple times within a scenario should be
avoided; instead, sample each event once and cache the result, which clarifies
the scientific model (each event occurs exactly once). Second, \emph{accidental
collision}: even for theoretically distinct events, a hash-based approach could
yield identical keys. With sufficiently large identifier spaces (e.g.\
128-bit), collision probability is negligible; duplicates can be asserted
against in debug builds.

Second, event identifiers should be \emph{isolated} from spurious inputs whose
values change across scenarios due to the intervention. Doing so makes the
identity of a transmission event's $U_e$ a descendant of upstream events,
reintroducing execution-history dependence in a new form. For example,
imagine including a summary of an endogenous model quantity (e.g. incidence on
day $t$) in the event key for an individual's infection event on day $t$. This
would cause that individual's residual noise to change across scenarios
whenever incidence differs. If history (e.g. recent population incidence) affects
risk (e.g. via behavior change), that dependence should enter through explicitly modeled state
$S_e$ and hence through $p_e(S_e)$, rather than implicitly through the
exogenous noise identity.

Third, a central problem specific to agent- or individual-based simulations is
that agent identifiers are often global, mutable stateful counters (not
dissimilar from the stateful PRNGs). Any stochastic birth-order changes (e.g. a
parent dies, so the child they produce in a vaccine intervention scenario is
never born in the baseline scenario) can cause the unique agent identifiers to
become mismatched across counterfactual worlds. Since agent identifiers are
used to key exogenous random noise, these must remain stable. Fortunately, the
solution mirrors event keying: the population indexes can be formed from
identifying information. The ``founder set'' might have indices $1, \ldots,
N_0$ across scenarios (perhaps differentiated by seed), and subsequent
offspring receive identifiers derived from their parent's identifier and birth
order, $\text{id}_{\text{child}} = h(\text{id}_{\text{parent}}, k)$ where $k$
indexes the $k$th offspring of that parent. This ensures that genealogically
corresponding agents share identifiers across scenarios, though alternative
keying strategies (e.g.\ incorporating birth timing) may be more appropriate
depending on the scientific model.

Finally, event-keyed randomness cleanly handles scenario-specific,
conditionally occurring events. Each conditionally-occurring event has a
well-defined key regardless of whether it occurs in a particular scenario. If
person 1 is not infected, the incubation key \texttt{hash(seed, "incubation",
1)} is simply never queried. Regardless, its existence does not shift any other
event's random draw as it would under a stateful PRNG. This separation between
\emph{defining} an event's noise identity and \emph{using} it mirrors the SCM
semantics where $\mathbf{U}$ contains noise for all possible events, only some
of which are realized in any given scenario.

% \section{Event Key Perturbation and Causality}
%
% With event-keyed randomness and deterministic state updates, changing the
% random choice at a single addressed event while holding all other addressed
% choices fixed produces differences that are entirely mediated by that event
% through the simulator’s own dependency structure; in particular, no unrelated
% downstream differences arise from call-order shifts.
%
% This is the simulator analogue of a surgical intervention on one noise variable
% in a structural causal model: all other exogenous noises are held fixed, so
% differences are attributable to the intervened noise via deterministic
% propagation.

\section{Discussion}
\label{sec:discussion}

We have argued that stateful PRNGs create a fundamental mismatch between the
scientific causal structure ABMs are intended to encode, and the program-level
causal structure they actually implement. This divergence leads to practical
limitations (e.g. variance reduction using CRNs may not be efficient), as well
as more serious scientific problems (e.g. individual-level counterfactual
estimates may be ill-posed). We frame these well-known issues with stateful
PRNGs in a structural causal model framework, which formalizes how these
algorithms lead to invalid counterfactuals. We argue that event-keyed
randomness, combining counter-based PRNGs with stable event identifiers,
restores the execution invariance required for SCM-valid counterfactual
comparisons.

In arguing for counter-based PRNGs, we have shown that they bring a new
requirement from modelers: since event-keys couple the exogenous noise for
events across counterfactual worlds, they require intentional design. This is a
fundamental modeling choice that is unavoidable, and not automatable since
event-keys correspond to different underlying counterfactual semantics that are
the responsibility of the researcher to fully specify. This choice also implies
a spectrum of event design corresponding with different model assumption
strengths. A simple event key implies a strong assumption of exchangeability
between events, while a sufficiently complex event key can approach independent-seed
comparisons with no shared events between scenarios.

On the practical side, modern counter-based PRNGs are comparable in speed to
stateful generators, while enabling the event-keyed access semantics our work
requires. In particular, Philox is about 2x slower and Threefry is about 5\%
faster than the Mersenne Twister \citep[Table 2]{Salmon2011-zz}, though these
comparisons may be degraded when accounting for complex event key calculation.
We note, however, that unlike stateful PRNGs, counter-based PRNGs are naturally
parallelizable, meaning that in multicore settings counter-based PRNGs may
unlock faster performance than stateful PRNGs.

More broadly, event-keyed randomness connects ABM methodology to functional
programming principles: making random draws pure functions of their inputs
rather than side effects of mutable state. This same property that enables
causal validity also facilitates parallelization, debugging, and
reproducibility more generally. We hope this work encourages the ABM community
to treat execution invariance not as an optimization detail but as a core
requirement for causally coherent simulation-based inference.

\begin{appendices}

\section{Consequences of Execution Invariance Violations}
\label{app:consequences}

Execution invariance violations undermine the use of ABMs to answer causal
questions in three ways. First, they make variance reduction from common random
numbers via fixed initial seed unpredictable, an issue pointed out by
\citep{klein2024noise}. Second, and more fundamentally, they render
individual-level counterfactuals ill-defined. Third, they compromise other
downstream analyses (e.g. sensitivity analyses, variance decomposition,
analyzing individual trajectories) that rely on stable noise assignments. We
discuss each of these issues in turn.

\paragraph{Statistical Efficiency.} The most immediate practical consequence of
execution invariance violations is that the variance reduction benefits of
using CRNs to estimate ATEs become unpredictable. When
interventions alter draw indices, the covariance $\Cov(Y^{(1)}, Y^{(0)})$
may even turn negative as noise terms stochastically misalign, in which case
seed-matching \emph{increases} variance rather than reducing it. At best,
misalignment produces weak or zero covariance, eliminating variance reduction
benefits.

\paragraph{Counterfactual Coherence.}
More fundamentally, execution invariance violations render individual-level
treatment effects ill-defined. In Pearl's SCM framework, a counterfactual
comparison asks: ``What would happen to this unit under a different
intervention, holding all else equal?'' --- where ``all else'' includes the
exogenous noise terms $\mathbf{U}$.

As we showed in Figure~\ref{fig:option2-scm}, stateful PRNGs introduce implicit
endogenous variables (the draw indices $K_e$) that make noise term identity
depend on execution history rather than event identity alone. This means
exogenous noise inputs are no longer held constant across scenarios. When we
ask ``what would person 2's infection status be if we vaccinated person 1?'',
we are asking about the same person 2 experiencing different causal conditions.
This requires the same exogenous noise $U_2$ in both scenarios. With stateful
PRNGs and seed-matching, person 2 receives $R_2$ in one scenario and $R_3$ in
another --- these represent different random events, not the same person under
different treatments.

Consequently, the individual treatment effects $\mathrm{ITE}_i = Y_i^{(1)} -
Y_i^{(0)}$ which we use to estimate the ATE become comparisons between
\emph{different chance events} rather than causal effects \emph{on the same
unit}. This violates the fundamental coupling between baseline and intervention
``worlds'' assumption underlying counterfactual reasoning. This is especially
troubling for ABMs, which are unique in their ability to simulate both
potential outcomes to overcome the fundamental problem of causal inference.

\paragraph{Auxiliary Analyses.} Often, researchers conduct exploratory or
post-model-fitting analyses to better understand the behavior of their
(sometimes rather complex) simulation model. For example, in a
\emph{sensitivity analysis}, a researcher perturbs a single parameter to see
its effects on outcomes (e.g. varying vaccine efficacy $VE$ in our toy model to
see how sensitive incidence is to this parameter). However, a consequence of
violations of execution invariance is that such parameter perturbations may
trigger a new random draw, which then breaks the downstream matching between
exogenous noise terms and modeled events. Under the intended scientific SCM,
minor parameter changes should not affect unrelated events' exogenous noise,
but stateful PRNGs violate this through draw-index coupling.

Another auxiliary analysis researchers often do in simulation-based work is
\emph{variance decompositions}, such as Sobol indices \citep{Sobol-2001-gt}.
These approaches partition outcome variance into contributions from different
inputs. Computing first-order Sobol indices $S_i = V_i / \Var(Y)$ requires
estimating $V_i = \Var_{X_i}[\E_{X_{\sim i}}(Y \mid X_i)]$ by fixing parameter
$X_i$ at various values while varying all other parameters $X_{\sim i}$ using
different random realizations. Execution invariance violations contaminate this
estimation: when changing $X_i$ alters execution paths, the random draws used
for $X_{\sim i}$ stochastically misalign across different fixed values of
$X_i$, introducing spurious variance that does not reflect the model's actual
sensitivity structure.

Additionally, \emph{mediation analysis} can be adversely affected by execution
invariance violations. Mediation analysis is a causal inference technique that
decomposes treatment effects into direct and indirect pathways by intervening
on intermediate variables \citep{Pearl2001-xf,VanderWeele2016-xy}. For example,
in transmission models vaccine effects can be decomposed into direct protection
of vaccinated individuals versus indirect protection of others through reduced
transmission. ABMs are well-suited to estimate such effects via simulation
interventions. Such decompositions require intervening on intermediate
variables (e.g., vaccine infection status) while holding all exogenous noise
fixed, which is precisely the coupling that execution invariance provides. The
spurious causal paths introduced by endogenous random draw indices
(Figure~\ref{fig:option2-scm}) create artificial dependencies between
mechanistically unrelated events, contaminating mediation estimates.

\end{appendices}

\section*{Acknowledgements}
We thank Kyra Grantz and Christine Wen for helpful feedback and discussions.

\bibliography{references}

\end{document}